\documentclass[fleqn,10pt]{wlscirep}
\usepackage{graphicx}
\usepackage{amsmath}
\usepackage{amssymb}
\usepackage{mathrsfs}
\usepackage{amsthm}
\usepackage{bm}
\usepackage{url}
\usepackage{csquotes}
\MakeOuterQuote{"}
\usepackage{exscale}
\usepackage{latexsym}
\usepackage{cases}

\newtheorem{theorem}{Theorem}
\newtheorem{lemma}{Lemma}

\newcommand{\Tr}{\operatorname{Tr}}

\newcommand{\ket}[1]{|#1\rangle}
\newcommand{\bra}[1]{\langle #1|}

\newcommand{\Cl}{\mathcal{C}_{l_1}}

\newcommand{\I}{\mathcal{I}}
\newcommand{\C}{\mathcal{C}}
\newcommand{\n}{\nonumber\\}

\newcommand{\M}{\mathcal{M}}

\title{Coherence transformations in single qubit systems}

\author[1,2]{Hai-Long Shi}
\author[2,3,*]{Xiao-Hui Wang}
\author[1,3,4]{Si-Yuan Liu}
\author[1,3]{Wen-Li Yang}
\author[2,3]{Zhan-Ying Yang}
\author[4,1,3]{Heng Fan}
\affil[1]{Institute of Modern Physics, Northwest University, Xi'an 710069, China}
\affil[2]{School of Physics, Northwest University, Xi'an 710069, China}
\affil[3]{Shaanxi Key Laboratory for Theoretical Physics Frontiers, Xi'an 710069, China}
\affil[4]{Institute of Physics, Chinese Academy of Sciences, Beijing 100190, China}
\affil[*]{xhwang@nwu.edu.cn}

\begin{abstract}
We investigate the single qubit transformations
under several typical coherence-free operations, such as,
incoherent operation (IO), strictly incoherent operation (SIO),
physically incoherent operation (PIO), and coherence-preserving operation (CPO).
Quantitative connection has been built between IO and SIO in single qubit systems.
Moreover, these coherence-free operations have a clear hierarchical relationship in single qubit systems:
CPO $\subset$ PIO $\subset$ SIO=IO. A new and explicit proof for the necessary and sufficient condition of single qubit transformation via IO or SIO
has been provided,
which indicates that SIO with only two Kraus operators are enough to realize this transformation.
The transformation regions of single qubits via CPO and PIO are also given.
Our method provides a geometric illustration to analyze single qubit coherence transformations
by introducing the Bloch sphere depiction of the transformation regions,
and tells us how to construct the corresponding coherence-free operations.
\end{abstract}

\begin{document}

\flushbottom
\maketitle

\thispagestyle{empty}

\section*{Introduction}

Quantum resource theory has become a powerful tool
in quantitatively describing
many intriguing and novel characteristics of quantum systems \cite{Gour}.
A general quantum resource theory includes two basic ingredients: ``free'' states and ``free'' quantum operations.
A major concern of any resource theory is how to quantify and manipulate these resource states, i.e., non-free states.
Much attention has been paid to this direction
\cite{Vedral,eanglement-measure,discord-measure1,discord-measure2,
Baumgratz, coherence-measure1, coherence-measure2, RDM, Guo,
entanglement-transformation, coherence-transformations, PIO, PIO-3, PIO-2}.
For instance, in the resource theory of entanglement,
the free operations are local quantum operations with classical communication (LOCC),
and possible entanglement manipulations between bipartite entangled states via LOCC are determined by majorization \cite{entanglement-transformation}.
Recently, quantum coherence, as another embodiment of quantum states superposition principle,
has received widespread attention and scrutiny since
it can be viewed as a vital quantum resource in various quantum information processes,
such as, quantum algorithms
\cite{quantum-algorithm1, quantum-algorithm2, quantum-algorithm3, quantum-algorithm4, quantum-algorithm5},
quantum metrology \cite{quantum-metrology1, quantum-metrology2},
and quantum channel discrimination \cite{channel-discrimination1, channel-discrimination2}.
Besides, many coherence-free operations have been proposed,
including incoherent operation (IO) \cite{Baumgratz}, strictly incoherent operation (SIO) \cite{SIO1, SIO2},
physically incoherent operation (PIO) \cite{PIO}, coherence-preserving operation (CPO) \cite{CPO},
and ``maximal'' incoherent operation (MIO) \cite{MIO}.
A natural question is how to utilize this precious quantum resource via coherence-free operations
for the realization of quantum state transformations.

In Ref. \citen{coherence-transformations}, it has been shown that
a pure state $\ket{\psi}$ can be transformed to another pure state $\ket{\phi}$ using IO
if and only if the square moduli of superposed coefficients $(|\psi_1|^2,\ldots,|\psi_d|^2)^{t}$ are majorized by
$(|\phi_1|^2,\ldots,|\phi_d|^2)^{t}$.
For the case of mixed state, Chitambar and Gour \cite{PIO, PIO-3, PIO-2}
considered the transformations of single qubit mixed states
and first obtained a necessary and sufficient condition
for single qubit transformations by either SIO, DIO, IO, or MIO.
The proof of this condition also tells us how to construct the corresponding SIO
for possible single qubit transformations.
However, this construction for realization of the single qubit transformation from $\rho$ to $\rho'$
needs an intermediate state $\rho'_{max}$,
i.e., $\rho\rightarrow\rho'_{max}\rightarrow\rho'$ \cite{PIO-2}.
Thus four Kraus operators are needed to construct a SIO for a direct transformation: $\rho\rightarrow\rho'$.
For this reason, we would like to provide a direct approach to complete transformation from $\rho$ to $\rho'$,
where less Kraus operators are needed.
In addition, we will use the Bloch sphere depiction of single qubit \cite{Nielsen}
to better illustrate and understand the coherence transformation of single qubit.

In this paper, we discuss how to implement single qubit transformations
via four kinds of incoherent operations,
namely, IO, SIO, PIO, and CPO.
Firstly, we use the Bloch sphere depiction to parameterize single qubit and
discover that the transformation ability of single qubit via four kinds of incoherent operations
has rotational symmetry around $z$-axis in the cylindrical coordinates,
which simplifies the following discussion.
Secondly, in single qubit systems, the relation between IO and SIO is IO=SIO,
which has been proposed by Chitambar and Gour \cite{PIO-3, PIO-2}.
Further, we build the quantitative connection between them in single qubit systems.
Then we offer a new method to construct the map for realization of
single qubit transformation via IO,
where the intermediate state $\rho'_{max}$ is no longer necessary
and only two special Kraus operators are needed.
One of them is represented by a diagonal matrix and the other is represented by an anti-diagonal matrix.
Additionally, by exploring these two special Kraus operators,
we provide a different and  explicit proof for the necessary and sufficient condition of single qubit transformation via IO.
The transformation regions of CPO, IO, PIO are also obtained in the Bloch sphere depiction.
Finally, we discuss two examples:
maximally coherent state transformations via IO
and pure state transformations via IO.
Our results offer new insight into the power of incoherent operations in quantum state manipulation
by introducing the Bloch sphere depiction of the transformation region.

\section*{Results}\label{sec2}
\textbf{{Definition.}-}\quad
To begin with, let us first give a brief review of several typical incoherent operations
and coherence measures.
In quantifying coherence \cite{Baumgratz},
a particular base $\{\ket{i}\}$ should be chosen and fixed.
The density operators of incoherent quantum states $\delta$ are diagonal in this base,
i.e., $\delta=\sum_ic_i\ket{i}\bra{i}$.
A set of these incoherent quantum states is labeled by $\mathcal{I}$,
and IO is denoted as $\Lambda^{IO}$,
where Kraus operators $\{K_n\}_{n=1}^r$ fulfil
\begin{equation}
\frac{K_n\delta K_n^{\dag}}{\Tr[K_n\delta K_n^{\dag}]}\in \mathcal{I}.
\end{equation}

\begin{lemma}\label{lem:1}\cite{Sun} There exists at most one nonzero entry in every column of
the Kraus operator $K_n$ belonging to $\Lambda^{IO}$.
\end{lemma}

According to Lemma~\ref{lem:1}, the Kraus operators of IO can be expressed as $K_n=\sum_{i=0}^{d-1}c_{ni}\ket{f_n(i)}\bra{i}$, $n=1\cdots r$,
where $f_n:\{0,\ldots, d-1\}\rightarrow\{0,\ldots, d-1\}$ and $d$ is the dimension of Hibert space.
An incoherent operation is called SIO if its $K_n$ also satisfies \cite{SIO1, SIO2}
\begin{equation}
\frac{K_n^{\dag}\delta K_n}{\Tr[K_n^{\dag}\delta K_n]}\in\I.
\end{equation}
Similarly, we can get the form of SIO that every column and row of its $K_n$ has at most one nonzero entry.

The CPO was introduced in Ref. \citen{CPO} to reveal that
coherence of a state is intrinsically hard to preserve when
there is a lack of information about the state and the quantum channel.
A unitary and incoherent operation is CPO,
which keeps the coherence of quantum states invariant,
i.e., $\C[\Lambda^{CPO}(\rho)]=\C(\rho)$ ($\C$ is a coherence measure).
Thus, the Kraus operator of CPO takes the following form \cite{CPO}:
\begin{equation}
K=\sum_ie^{i\theta_i}\ket{\pi(i)}\bra{i},
\end{equation}
where $\pi$ is a permutation.
Note that a CPO belongs to a class of IO with only one Kraus operator
due to $\sum_nK_n^\dag K_n=I$ .

To establish a physically consistent resource theory,
the PIO was proposed to replace IO in quantifying coherence \cite{PIO}.
Since a set of Kraus operators can be physically realized by
introducing auxiliary particles and making appropriate unitary operations and projective measurement,
a PIO requires that they are all incoherent.
Following this ideal, the expression of PIO has been obtained in Ref. \citen{PIO}.
The PIO can be expressed as a convex combination of maps,
which have Kraus operators $\{K_n\}_{n=1}^r$ of the form:
\begin{equation}\label{PIO}
K_n=U_nP_n=\sum_ie^{i\theta_{ni}}\ket{\pi_n(i)}\bra{i}P_n,
\end{equation}
where the $P_n$ form an orthogonal and complete set of incoherent projectors.
Hence, these incoherent operations have a clear hierarchical relationship:
CPO $\subset$ PIO $\subset$ SIO $\subset$ IO.

The first rigorous framework of quantifying coherence was proposed in Ref. \citen{Baumgratz},
where a function $\C$ can be taken as a coherence measure if it satisfies the following conditions \cite{Baumgratz}:
\\(B1) $\C(\rho)\geq0$ for all quantum states and $\C(\rho)=0$ if and only if $\rho\in\mathcal{I}$;
\\(B2) $\C(\rho)\geq\sum_n p_n\C(\rho_n)$, where $p_n=\Tr(K_n\rho K_n^\dag)$, $\rho_n=K_n\rho K_n^\dag/p_n$,
and $K_n$ are the Kraus operators of IO;
\\(B2') $\C(\rho)\geq\C[\Lambda^{IO}(\rho)]$; and
\\(B3) $\sum_n p_n\C(\rho_n)\geq\C(\sum_n p_n \rho_n)$ with $p_n\geq0$ and $\sum_np_n=1$.
On the basis of this framework, the relative entropy of coherence
and $l_1$ norm of coherence were put forward to measure coherence degree of quantum states.
The $l_1$ norm of coherence is defined as \cite{Baumgratz}
\begin{equation}\label{cl1}
\Cl(\rho)=\sum_{i\neq j}|\rho_{ij}|,
\end{equation}
which comes from a simple fact that
coherence is linked with the off-diagonal elements of considered quantum states.
\\
\\
\textbf{{Relation between IO and SIO.}-}\quad
In the cylindrical coordinates, density matrices of single qubit systems can be written as
\begin{equation}
\rho=\frac{1}{2}\begin{pmatrix} 1+z & re^{-i\theta} \\   re^{i\theta} & 1-z \\ \end{pmatrix},
\end{equation}
where $-1\leq z\leq1$, $0\leq r\leq1$, and $0\leq\theta\leq\pi$.
We first prove the following Lemma~\ref{lem:2} to simplify our discussion.

\begin{lemma}\label{lem:2} $\rho_2=\Lambda(\rho_1)$ if and only if $\tilde{\rho_2}=\tilde{\Lambda}(\widetilde{\rho_1})$ where $\Lambda$ and $\tilde{\Lambda}$ are IO, and
\begin{equation}
\tilde{\rho}=\frac{1}{2}\begin{pmatrix} 1+z & r \\ r & 1-z \\ \end{pmatrix}.
\end{equation}
\end{lemma}

\begin{proof} It is clear that $\rho=U\tilde{\rho}U^{\dag}$ with $U=$diag($e^{-i\theta/2}$, $e^{i\theta/2}$).
If $\rho_2=\Lambda(\rho_1)$ then we have
\begin{eqnarray}
\tilde{\rho_2}&=&U_2^\dag\Lambda(U_1\tilde{\rho_1}U_1^{\dag})U_2
=\sum_nU_2^\dag K_nU_1\tilde{\rho_1}U_1^{\dag}K_n^{\dag}U_2.
\end{eqnarray}
Let $\tilde{K_n}=U_2^\dag K_nU_1$. It is easy to check that $\sum_n\tilde{K_n}^\dag\tilde{K_n}=I$.
Now let us show that $\tilde{K_n}$ is also incoherent.
Suppose $K_n=\sum_ic_{ni}\ket{f_n(i)}\bra{i}$ then we have
\begin{eqnarray}
\tilde{K_n}=U_2^\dag K_nU_1&=&\sum_{ijk}u^{(i)*}_2\ket{i}\bra{i}c_{nj}\ket{f_n(j)}\bra{j}u^{(k)}_1\ket{k}\bra{k}
=\sum_{k}u^{[f_n(k)]*}_2u^{(k)}_1c_{nk}\ket{f_n(k)}\bra{k},
\end{eqnarray}
which means that $\tilde{K_n}$ is also incoherent.
By using the same approach, we can prove that there exists an IO making $\rho_2=\Lambda(\rho_1)$ when $\tilde{\rho_2}=\tilde{\Lambda}(\widetilde{\rho_1})$.
\end{proof}

Lemma~\ref{lem:2} also holds for SIO, PIO, or CPO.
This lemma implies that the coherence transformation ability of single qubit
is depended only on two parameters ($z,r$)
and not on the parameter $\theta$,
i.e., rotational symmetry around z-axis.
Therefore, we only need to consider the coherence transformations between the quantum states of $\tilde{\rho}$.
In the following text, we use symbol $\rho$ to represent $\tilde{\rho}$ for convenience.
Meanwhile, we denote initial qubit $\rho$ by $(z,$ $r)$ and represent
transformation region $\rho'$ of the initial qubit $\rho$ via coherence-free operations by $(z',$ $r')$.
With these notions, we prove the following theorem.

\begin{theorem}\label{thm:1}
In single qubit systems, the transformation region given by IO
is equal to the transformation region given by SIO.
\end{theorem}

\begin{proof}
Define four types of Kraus operators as follows
\begin{eqnarray}
\mathcal{M}_1=\begin{pmatrix} \times&\times\\0&0\end{pmatrix}
,\quad\mathcal{M}_2=\begin{pmatrix} 0&0\\ \times&\times\end{pmatrix},\n
\mathcal{M}_3=\begin{pmatrix} \times&0\\0&\times\end{pmatrix}
,\quad\mathcal{M}_4=\begin{pmatrix} 0&\times\\ \times&0\end{pmatrix},
\end{eqnarray}
where "$\times$" means that the elements of matrix may not equal to zero.
The above four types of Kraus operators depict all IO applied in single qubit transformations
and the maps whose Kraus operators belonging to $\M_3$ or $\M_4$ are SIO.

Suppose that we have any IO represented by a set of Kraus operators $\Lambda^{IO}=\{K_i,$ $K_j,$ $K_l\}$
where
\begin{equation}
K_i=\begin{pmatrix}A_i&B_i\\0&0\end{pmatrix},\quad K_j=\begin{pmatrix}0&0\\C_j&D_j\end{pmatrix},
\end{equation}
and $K_l\in\M_3\cup\M_4$.
Next we would like to replace $\Lambda^{IO}$ with $\Lambda^{SIO}$
while keeping $\Lambda^{SIO}(\rho)=\Lambda^{IO}(\rho)$.
Here, the SIO is in the form of $\Lambda^{SIO}=\{K_0,$ $K_1,$ $K_l\}$
and $K_0$, $K_1\in\M_3\cup\M_4$.
Define
\begin{equation}
K_0=\begin{pmatrix}a&0\\0&b\end{pmatrix}\quad \mathrm{and}\quad K_1=\begin{pmatrix}0&d\\c&0\end{pmatrix}.
\end{equation}
Now we prove that there exist $a$, $b$, $c$, and $d$ making
\begin{eqnarray}\label{SIO=IO}
\left\{
  \begin{array}{ll}
  \Lambda^{SIO}(\rho)=\Lambda^{IO}(\rho);\\
     K_0^\dag K_0+K_1^\dag K_1+\sum_l K_l^\dag K_l=I.
  \end{array}
\right.
\end{eqnarray}
By using the relationship $\sum_iK_i^\dag K_i+\sum_jK_j^\dag K_j+\sum_l K_l^\dag K_l=I$,
Eq. (\ref{SIO=IO}) reduces to
\begin{subequations}
\begin{numcases}{}
|a|^2+|c|^2=|A|^2+|C|^2;\\
|b|^2+|d|^2=|B|^2+|D|^2;\\
|a|^2(1+z)+|d|^2(1-z)=h_1;\\
\label{ab*+c*d}ab^*+c^*d=0,
\end{numcases}
\end{subequations}
where $h_1=|A|^2(1+z)+r\sum_i(B_iA_i^*+A_iB_i^*)+|B|^2(1-z)$,
$|A|^2=\sum_i|A_i|^2$, $|B|^2=\sum_i|B_i|^2$, $|C|^2=\sum_j|C_j|^2$,
and $|D|^2=\sum_j|D_j|^2$.
The Eq. (\ref{ab*+c*d}) can be rewritten as
\begin{equation}
|a|^2|b|^2=|c|^2|d|^2.
\end{equation}
since we can choose suitable phases for $a$, $b$, $c$, and $d$ to satisfy Eq. (\ref{ab*+c*d}).
Solving it we obtain
\begin{eqnarray}\label{solution-IO=SIO}
\left\{
\begin{array}{rcl}
|a|^2&=&(|A|^2+|C|^2)\frac{h_1}{h_1+h_2};\\
|b|^2&=&\frac{h_2}{1-z}-\frac{(1+z)(|A|^2+|C|^2)h_2}{(1-z)(h_1+h_2)};\\
|c|^2&=&(|A|^2+|C|^2)\frac{h_2}{h_1+h_2};\\
|d|^2&=&\frac{h_1}{1-z}-\frac{(1+z)(|A|^2+|C|^2)h_1}{(1-z)(h_1+h_2)},
\end{array}
\right.
\end{eqnarray}
where $h_2=|C|^2(1+z)+r\sum_j(D_jC_j^*+C_jD_j^*)+|D|^2(1-z)$.
Note that the solutions: $|a|^2$, $|b|^2$, $|c|^2$, and $|d|^2$ in Eq. (\ref{solution-IO=SIO})
may be negative.
Therefore, if we prove that they are always non-negative, then we can find SIO to replace IO.
Clearly, $h_1$ and $h_2$ are non-negative
due to $h_1=2\sum_i\Tr (K_i\rho K_i^\dag)$ and $h_2=2\sum_j\Tr (K_j\rho K_j^\dag)$.
Hence, the $|a|^2$, $|b|^2$, $|c|^2$, and $|d|^2$ of Eq. (\ref{solution-IO=SIO}) are non-negative.
\end{proof}

In Ref. \citen{PIO-2}, the authors have proved this result IO=SIO in single qubit systems
by the following two arguments: SIO $\subset$ IO $\subset$ MIO and MIO=SIO.
We provide a new and direct proof
and establish a quantitative correspondence between IO and SIO
in coherence transformations of single qubit systems.
By using Eq. (\ref{solution-IO=SIO}),
we can accurately construct a SIO to realize the role (quantum state transformations) of IO
in single qubit systems.
\\
\\
\textbf{{The transformation region given by CPO.}-}\quad
In the case of IO with only one Kraus operator $K$,
the $K$ must be unitary.
Hence, the $K$ also describe a CPO,
which can be expressed as $K=\sum_ie^{i\theta_i}\ket{\pi(i)}\bra{i}$.
For single qubit systems, the Karus operator of CPO has two forms:

$\bullet$ Case 1: $K=e^{i\theta_1}\ket{0}\bra{0}+e^{i\theta_2}\ket{1}\bra{1}$.
By using this type of CPO, the transformable quantum states are
\begin{eqnarray}
K\rho K^{\dag}=\frac{1}{2}\begin{pmatrix} 1+z && re^{i(\theta_1-\theta_2)} \\ re^{i(\theta_2-\theta_1)} && 1-z \end{pmatrix},
\end{eqnarray}
where initial state is $\rho=\frac{1}{2}\begin{pmatrix} 1+z && r \\ r && 1-z \\ \end{pmatrix}$.
We only need to consider quantum states in the form of real parameters due to Lemma~\ref{lem:1}.
Therefore, the transformable quantum states are $(z,$ $r)$ and $(z,$ $-r)$.

$\bullet$ Case 2: $K=e^{i\theta_1}\ket{0}\bra{1}+e^{i\theta_2}\ket{1}\bra{0}$. We have
\begin{eqnarray}
K\rho K^{\dag}=\frac{1}{2}\begin{pmatrix} 1-z && re^{i(\theta_1-\theta_2)} \\ re^{i(\theta_2-\theta_1)} && 1+z \end{pmatrix}.
\end{eqnarray}
The same procedure is easily adapted to obtain the transformable quantum states, $(-z,$ $r)$ and $(-z,$ $-r)$, under this kind of CPO.

By using CPO, the initial quantum state $(z,$ $r)$
can be transformed to $(z,$ $\pm r)$ and $(-z,$ $\pm r)$ (see Fig. 1).
Besides, these transformations between four quantum states are reversible.
\begin{figure}[ht]
 \centering
 \includegraphics[height=6.5cm]{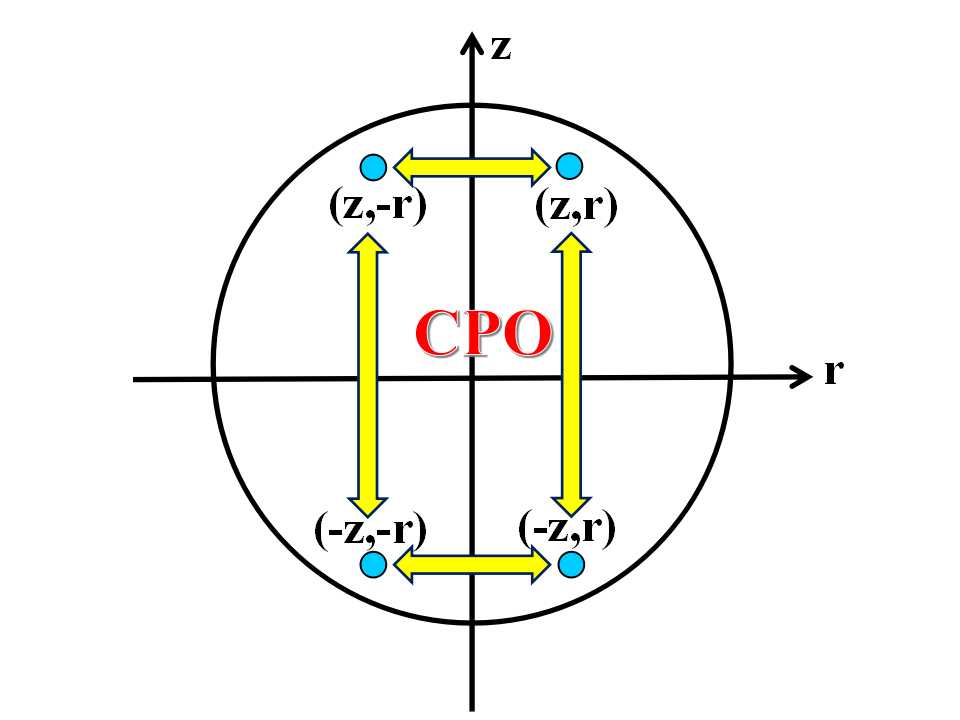}
 \caption{(Color online) Single qubit transformations under CPO, i.e., IO with only one Kraus operator.
  The initial quantum state is $(z,$ $r)$, and transformation regions are $(z,$ $\pm r)$ and $(-z,$ $\pm r)$.
  Particularly, these transformations are reversible.}
 \label{fig:algorithm}
\end{figure}
\\
\\
\textbf{{The transformation region given by IO.}-}\quad
In this section, we will construct a special IO with only two Kraus operators
belonging to $\mathcal{M}_3$ and $\mathcal{M}_4$, respectively.
From this case, we will get a transformation region of single qubit under IO,
and then we will prove it is also a maximal transformation region in the section of Methods.

Now we consider a special IO in the form of
\begin{eqnarray}\label{Kraus-special}
K_0=c_{00}\ket{0}\bra{0}+c_{11}\ket{1}\bra{1},\n
K_1=c_{10}\ket{1}\bra{0}+c_{01}\ket{0}\bra{1}.
\end{eqnarray}
According to Lemma~\ref{lem:1}, the above Kraus operators are incoherent.
Substituting the Eq. (\ref{Kraus-special}) in to $\sum_iK_n^\dag K_n=I$, we obtain
\begin{eqnarray}
|c_{00}|^2+|c_{10}|^2=1,\n
|c_{11}|^2+|c_{01}|^2=1.
\end{eqnarray}
We suppose that $c_{00},c_{01},c_{10},c_{11}\in\mathbb{R}$
and consider the following cases:
\\Case 1: $c_{00}=\sqrt{\alpha}$, $c_{10}=\sqrt{1-\alpha}$, $c_{11}=\sqrt{\beta}$ and $c_{01}=\sqrt{1-\beta}$;
\\Case 2: $c_{00}=\sqrt{\alpha}$, $c_{10}=-\sqrt{1-\alpha}$, $c_{11}=\sqrt{\beta}$ and $c_{01}=\sqrt{1-\beta}$;
\\Case 3: $c_{00}=-\sqrt{\alpha}$, $c_{10}=\sqrt{1-\alpha}$, $c_{11}=\sqrt{\beta}$ and $c_{01}=\sqrt{1-\beta}$;
\\Case 4: $c_{00}=-\sqrt{\alpha}$, $c_{10}=-\sqrt{1-\alpha}$, $c_{11}=\sqrt{\beta}$ and $c_{01}=\sqrt{1-\beta}$.
\\The qubit $\rho=\frac{1}{2}\begin{pmatrix} 1+z && r \\ r && 1-z \\ \end{pmatrix}$ after this type of IO becomes
\begin{equation}
\Lambda^{IO}(\rho)=\frac{1}{2}
\begin{pmatrix} 1+z' && r' \\ r' && 1-z' \\ \end{pmatrix},
\end{equation}
where $1+z'=\alpha(1+z)+(1-\beta)(1-z)$ and $r'=\lambda r$ with $\lambda=\sqrt{\alpha\beta}+\sqrt{(1-\alpha)(1-\beta)}$ in case 1.
In case 2, $\lambda=\sqrt{\alpha\beta}-\sqrt{(1-\alpha)(1-\beta)}$.
In case 3, $\lambda=-\sqrt{\alpha\beta}+\sqrt{(1-\alpha)(1-\beta)}$.
In case 4, $\lambda=-\sqrt{\alpha\beta}-\sqrt{(1-\alpha)(1-\beta)}$.
Note that
\begin{equation}
|\lambda|\leq\sqrt{\alpha\beta}+\sqrt{(1-\alpha)(1-\beta)}\leq1.
\end{equation}
Therefore,
\begin{equation}\label{r'-region}
|r'|\leq|r|.
\end{equation}

Setting $\tilde{\alpha}=\frac{1}{\sqrt{2}}(\alpha+\beta-1)$ and $\tilde{\beta}=\frac{1}{\sqrt{2}}(\alpha-\beta)$, then we have
\begin{equation}\label{lambda-equation}
\frac{2}{\lambda^2}\tilde{\alpha}^2+\frac{2}{1-\lambda^2}\tilde{\beta}^2=1,
\end{equation}
where case 1 corresponds to $\sqrt{\alpha\beta}\leq\lambda$ and $(\lambda^2+\alpha+\beta-1)/\lambda\geq0$;
case 2 corresponds to $\sqrt{\alpha\beta}\geq\lambda$ and $(\lambda^2+\alpha+\beta-1)/\lambda\geq0$;
case 3 corresponds to $\sqrt{\alpha\beta}\geq-\lambda$ and $(\lambda^2+\alpha+\beta-1)/\lambda\leq0$;
and, case 4 corresponds to $\sqrt{\alpha\beta}\leq-\lambda$ and $(\lambda^2+\alpha+\beta-1)/\lambda\leq0$.
According to Eq. (\ref{lambda-equation}), $\tilde{\alpha}$ and $\tilde{\beta}$ can be parameterized via $0\leq\theta\leq2\pi$
in the form of $\tilde{\alpha}=\sin\theta\lambda/\sqrt{2}$ and $\tilde{\beta}=\cos\theta\sqrt{(1-\lambda^2)/2}$.
The $z'$ expressed by $\theta$ is
\begin{eqnarray}\label{step2}
z'=\sqrt{(\lambda z)^2+1-\lambda^2}\sin(\theta+\phi),
\end{eqnarray}
where
\begin{eqnarray}\label{step1}
\cos\phi=\frac{\lambda z}{\sqrt{(\lambda z)^2+1-\lambda^2}},\n
\sin\phi=\sqrt{\frac{1-\lambda^2}{(\lambda z)^2+1-\lambda^2}}.
\end{eqnarray}
Above equation implies that
\begin{equation}\label{z'-region}
-\sqrt{(\lambda z)^2+1-\lambda^2}\leq z'\leq\sqrt{(\lambda z)^2+1-\lambda^2},
\end{equation}
whose boundary is an ellipse
\begin{equation}
\frac{z'^2}{1}+(1-z^2)\frac{r'^2}{r^2}=1.
\end{equation}

According to Eq. (\ref{r'-region}) and Eq. (\ref{z'-region}),
we obtain the transformation region $(z',r')$
\begin{eqnarray}\label{main result}
\left\{
  \begin{array}{ll}
    \frac{z'^2}{1}+(1-z^2)\frac{r'^2}{r^2}\leq1,\\
    |r'|\leq |r|,
  \end{array}
\right.
\end{eqnarray}
by using this special IO (see Fig. 2),
where $(z,r)$ represents the initial quantum states.
\begin{figure}[ht]
 \centering
 \includegraphics[height=6cm]{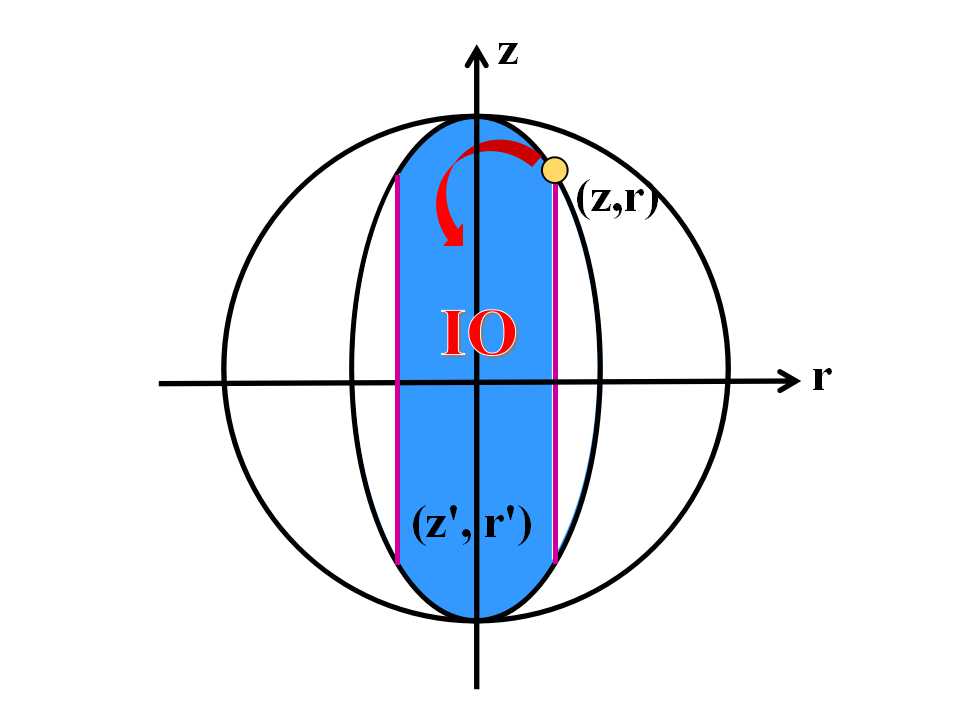}
 \caption{(Color online)  The transformation region of single qubit by IO or SIO is depicted by blue.
 The absolute value $|r|$ of purple lines is the $l_1$ norm of coherence of the initial state $(z,r)$.}
\end{figure}

\begin{theorem}\label{thm:2}
In single qubit systems, the region given by Eq. (\ref{main result})
is the maximal transformation region of the initila state $(z,$ $r)$ by using IO or SIO.
\end{theorem}

In the section of Methods, we will provide a complete proof of Theorem~\ref{thm:2}.
Theorem~\ref{thm:2} suggests that only two Kraus operators, which have the form of Eq. (\ref{Kraus-special}),
can describe all IO completely in single qubit systems.
Calculating the $l_1$ norm of coherence for single qubit systems via Eq. (\ref{cl1}),
we have
\begin{equation}\label{cl1'}
\Cl(\rho)=\sum_{i\neq j}|\rho_{ij}|=|r|,
\end{equation}
which is the boundary of transformation region (purple lines in Fig. 2).
It is consistent with the condition (B2') that the coherence of quantum states should not increase under IO.
Note that Theorem~\ref{thm:2} is also a necessary and sufficient condition to judge
whether a qubit can be transformed to another qubit via IO.
By using robustness of coherence and $\Delta$ robustness of coherence,
Ref. \citen{PIO, PIO-2, PIO-3} also provide a necessary and sufficient condition for single qubit transformations via IO,
which is consistent with our Eq. (\ref{main result}).
\\
\\
\textbf{{The transformation region given by PIO.}-}\quad
According to Eq. (\ref{PIO}), for any given orthogonal and complete set of incoherent projectors
($\{P_0=\ket{0}\bra{0}, P_1=\ket{1}\bra{1}\}$ or $\{P_0=I\}$),
the Kraus operators of single qubit systems have the following forms:

\begin{subequations}

\begin{equation}\label{PIO1}
\mathcal{K}_1=\left\{
K_0=\begin{pmatrix}
e^{i\theta_{00}}&0\\0&0
\end{pmatrix},
K_1=\begin{pmatrix}
0&0\\0&e^{i\theta_{11}}
\end{pmatrix}
\right\},\quad
\mathcal{K}_2=\left\{
K_0=\begin{pmatrix}
0&0\\e^{i\theta_{00}}&0
\end{pmatrix},
K_1=\begin{pmatrix}
0&e^{i\theta_{11}}\\0&0
\end{pmatrix}
\right\},
\end{equation}

\begin{equation}\label{PIO3}
\mathcal{K}_3=\left\{
K_0=\begin{pmatrix}
e^{i\theta_{00}}&0\\0&0
\end{pmatrix},
K_1=\begin{pmatrix}
0&e^{i\theta_{11}}\\0&0
\end{pmatrix}
\right\},\quad
\mathcal{K}_4=\left\{
K_0=\begin{pmatrix}
0&0\\e^{i\theta_{00}}&0
\end{pmatrix},
K_1=\begin{pmatrix}
0&0\\0&e^{i\theta_{11}}
\end{pmatrix}
\right\},
\end{equation}

\begin{equation}\label{PIO5}
\mathcal{K}_5=\left\{K=\begin{pmatrix}
e^{i\theta_{00}}&0\\0&e^{i\theta_{01}}
\end{pmatrix}\right\}
\quad \mathrm{or} \quad
\mathcal{K}_6=\left\{K=\begin{pmatrix}
0&e^{i\theta_{01}}\\e^{i\theta_{00}}&0
\end{pmatrix}\right\}.
\end{equation}
\end{subequations}
The PIO with Kraus operators of Eqs. (\ref{PIO1}) or (\ref{PIO3})
are coherence-breaking channels \cite{CBC},
and the PIO with Kraus operators of Eq. (\ref{PIO5}) are CPO.
The transformable quantum states $\rho'$ by using PIO are
\begin{equation}\label{PIO-region}
\rho'=\Lambda^{PIO}(\rho)=\sum_{i=1}^{6}p_i\Lambda_i^{PIO}(\rho),
\end{equation}
due to Eq. (\ref{PIO}),
where $\Lambda_i^{PIO}(\rho)=\sum_{K_n\in\mathcal{K}_i}K_n\rho K_n^\dag$,
$p_i\geq0$, $\sum_ip_i=1$, and $\rho$ is initial quantum state $(z,r)$.
It is easy to check that $\Lambda_i^{PIO}(\rho)$ ($i=1\cdots6$)
are $(z,\pm r)$, $(-z,\pm r)$, $(\pm z,0)$, and $(\pm1,0)$ in the Bloch sphere representation.
Therefore,
the transformation region of single qubit states via PIO is a convex hexagon
with six vertexes: $(z,\pm r)$, $(-z,\pm r)$, and $(\pm1,0)$,
which is depicted by blue region in Fig. 3.
\begin{figure}[ht]
 \centering
 \includegraphics[height=6cm]{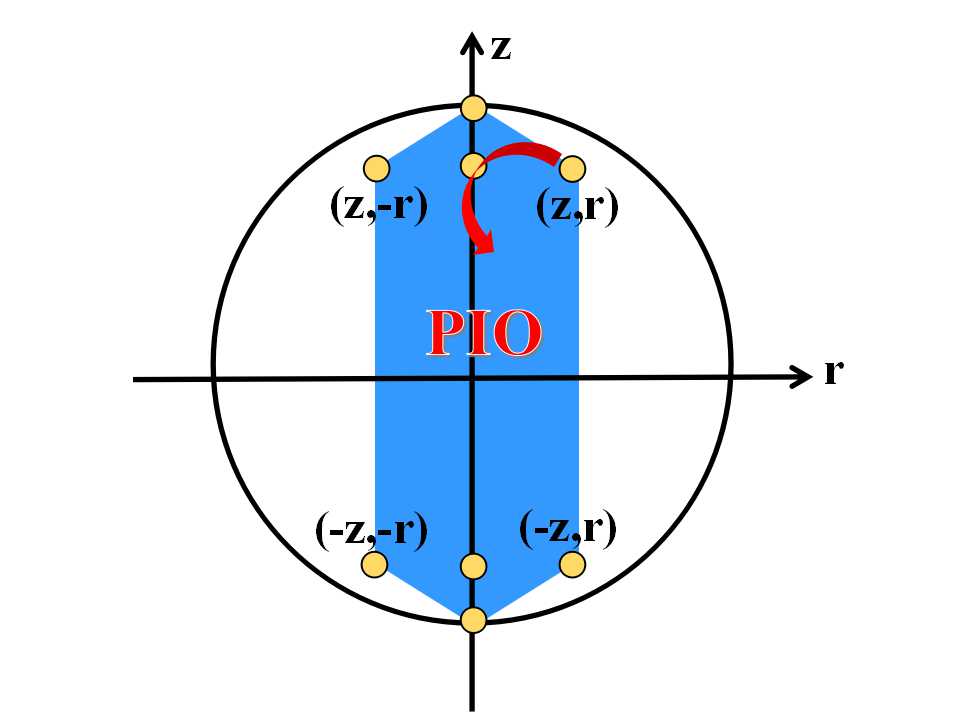}
 \caption{(Color online)  Single qubit transformations under PIO.
  The initial quantum state $\rho$ is $(z,r)$, and $\Lambda_i^{PIO}(\rho)$ are depicted by yellow points.
  The transformation region is represented by blue region.
  }
  \end{figure}

By introducing the Bloch sphere depiction of the transformation region, we can see that the coherence-free operations have a clear hierarchical relationship in single qubit systems:
CPO $\subset$ PIO $\subset$ SIO=IO; see Fig. 4.
\begin{figure}[ht]
 \centering
 \includegraphics[height=5cm]{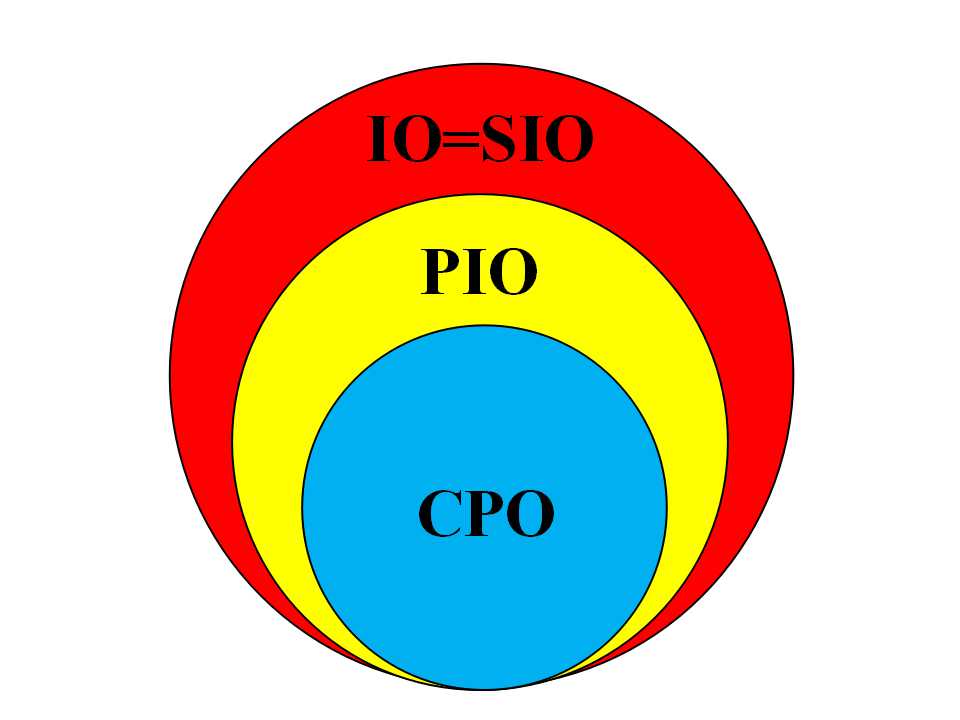}
 \caption{(Color online) The hierarchical structure of IO, SIO, PIO, and CPO in single qubit systems.}
 \label{fig:algorithm}
\end{figure}
\\
\\
\textbf{Example1-Maximally coherent state transformations via IO.}\quad
In Ref. \citen{Baumgratz}, Baumgratz \emph{et. al.} firstly found that
a d-dimensional maximally coherent state can be transformed to all other d-dimensional quantum states by means of IO.
However, the transformation in the proof of Ref. \citen{Baumgratz} is probabilistic.
Hence, how to prove that
a maximally coherent state allows for the deterministic generation of all other quantum states
is still an open question.
Here, we prove it in the case of single qubit systems.
In our notation, the maximally coherent state is denoted by $(z=0,$ $r=\pm1)$.
According to Eq. (\ref{main result}), the transformation region of maximally coherent state is
\begin{eqnarray}
\left\{
  \begin{array}{ll}
    z'^2+r'^2\leq1;\\
    |r'|\leq 1,
  \end{array}
\right.
\end{eqnarray}
which contains all single qubits (see Fig. 5).
Therefore, any single qubit can be determinately generated by a maximally coherent state by using IO.
\begin{figure}[ht]
 \centering
 \includegraphics[height=6cm]{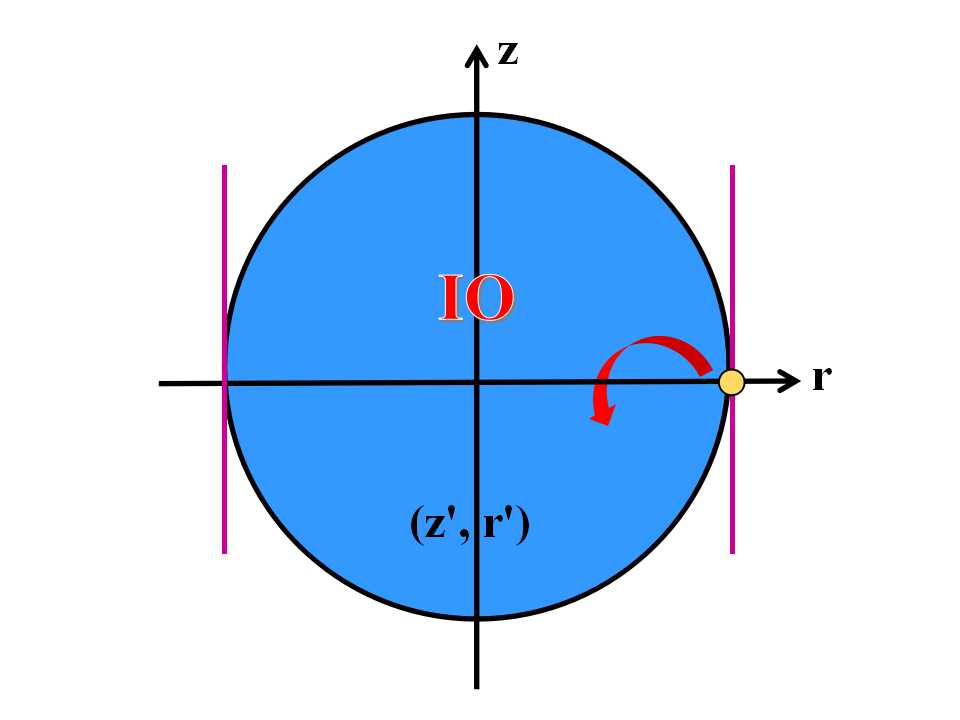}
 \caption{(Color online)  The transformation region given by IO is depicted by blue.}
\end{figure}

Now we construct the corresponding IO for a target quantum state $(z'=1/2,$ $r'=1/2)$ as an example
and $(z=0,$ $r=1)$ is chosen as the initial quantum state.
By virtue Eqs. (\ref{step2}) and (\ref{step1}), we obtain
\begin{equation}
\cos\theta=\frac{z'}{\sqrt{1-r'^2}}.
\end{equation}
Thus, $\tilde{\alpha}=r'\sqrt{(1-r'^2-z'^2)/(2-2r'^2)}=1/(2\sqrt{3})$
and $\tilde{\beta}=z'/\sqrt{2}=1/(2\sqrt{2})$.
Since $\alpha=(1+\sqrt{2}\tilde{\alpha}+\sqrt{2}\tilde{\beta})/2$
and $\alpha=(1+\sqrt{2}\tilde{\alpha}-\sqrt{2}\tilde{\beta})/2$,
we have $\alpha=3/4+1/(2\sqrt{6})$ and $\beta=1/4+1/(2\sqrt{6})$.
Due to $\sqrt{\alpha\beta}=\sqrt{(11+4\sqrt{6})/12}/2\geq\lambda=1/2$
and $(\lambda^2+\alpha+\beta-1)/\lambda=1/2+2/\sqrt{6}\geq0$,
we choose case 2 to construct Kraus operators and IO is
\begin{eqnarray}
\Lambda^{IO}=
\left\{\begin{pmatrix}\sqrt{\frac{3}{4}+\frac{1}{2\sqrt{6}}} &0\\0&\sqrt{\frac{1}{4}+\frac{1}{2\sqrt{6}}}
\end{pmatrix},\right.\n
\left.\begin{pmatrix}0&\sqrt{\frac{3}{4}-\frac{1}{2\sqrt{6}}}\\ -\sqrt{\frac{1}{4}-\frac{1}{2\sqrt{6}}}&0
\end{pmatrix}
\right\}.
\end{eqnarray}
\\
\\
\textbf{Example2-Pure state transformations via IO.}\quad
By using the Bloch sphere depiction of the transformation region,
one can see clearly that
$\ket{\psi}$ denoted by $(z=\sqrt{1-r^2},$ $r)$
transforms to $\ket{\phi}$ denoted by $(z'=\sqrt{1-r'^2},$ $r')$
using IO if and only if $\Cl(\ket{\psi})\geq\Cl(\ket{\phi})$ (see Fig. 6).
\begin{figure}
 \centering
 \includegraphics[height=6cm]{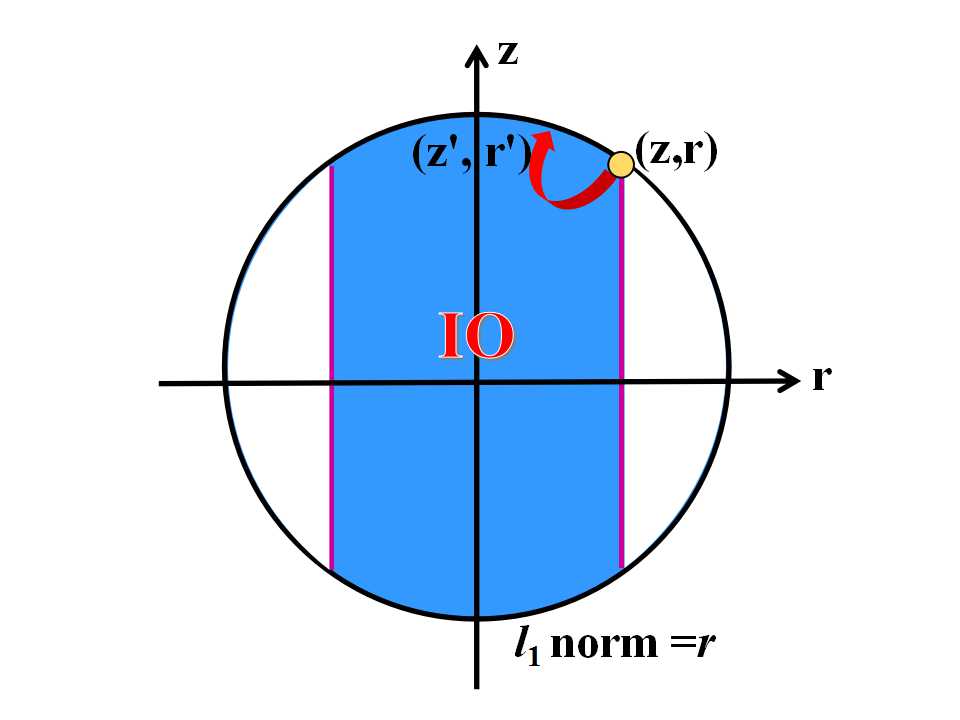}
 \caption{(Color online)  The transformation region given by IO is depicted by blue.}
\end{figure}
Similarly, we can also construct the corresponding IO for
$(z=1/\sqrt{3},$ $r=\sqrt{2/3})$ and $(z'=1/\sqrt{2},$ $r'=1/\sqrt{2})$
as a example, by using
Eqs. (\ref{Kraus-special}), (\ref{step2}), and (\ref{step1}).
The IO is
\begin{eqnarray}
\Lambda^{IO}=
\left\{\begin{pmatrix}\sqrt{\frac{1}2+\frac{\sqrt6}8+\frac{\sqrt2}8} &0\\0&\sqrt{\frac{1}2+\frac{\sqrt6}8-\frac{\sqrt2}8}
\end{pmatrix},\right.\n
\left.\begin{pmatrix}0&\sqrt{\frac{1}2-\frac{\sqrt6}8+\frac{\sqrt2}8}\\ \sqrt{\frac{1}2-\frac{\sqrt6}8-\frac{\sqrt2}8}&0
\end{pmatrix}
\right\}.
\end{eqnarray}

\section*{Discussion}\label{sec6}
In this paper, we have systematically studied the single qubit transformations
under IO, SIO, CPO, and PIO.
By introducing the Bloch sphere depiction,
we show that the transformation ability of single qubit via IO, SIO, CPO or PIO
has rotational symmetry around z-axis.
A quantitative correspondence between IO and SIO in single qubit systems has been established via Eq. (\ref{solution-IO=SIO}).
Therefore, we can concretely construct a SIO to replace a IO in single qubit transformations,
while keeping the initial and final states unchanged.
In the discussion of single qubit transformation via IO,
we provide a new and direct approach to obtain the necessary and sufficient condition.
The maximally single transformation region given by IO is depicted,
whose boundary is limited by the coherence value ($\Cl$) of initial state.
Our proof indicates that we can use a kind of special operation,
SIO with only two Kraus, to realize
all possible single qubit transformations given by IO or SIO.
And these special operations can be accurately constructed.
One of its Kraus operators is represented by a diagonal matrix,
and the other is represented by anti-diagonal matrix.
Finally, by calculating the transformation regions given by the above four operations,
we can understand the hierarchical relationship:
CPO $\subset$ PIO $\subset$ SIO=IO in single qubit systems
more directly.

An interesting question is whether
the transformation region of an initial qubit given by IO
can be defined as a coherence measure (denoted as $\mathcal{C}_a$) for single qubit systems.
From Fig. 2 and Theorem~\ref{thm:2}, we can see clearly that $\mathcal{C}_a$ fulfils conditions (B1) and (B2').
Other conditions, (B2) and (B3), for quantifying a suitable coherence measure need to be explored further.
Our results lead to an easy-operated and visual geometric
method to explore the power of coherence-free operations in single qubit manipulation,
and is worth applying to
investigate coherence transformations in multi-particle systems.

\section*{Methods}
\textbf{Proof of Theorem 2.}\quad
According to Theorem~\ref{thm:1}, IO can be expressed as $\Lambda^{IO}=\{K_i,K_j\}$,
where
\begin{equation}
K_i=
\begin{pmatrix}
a_i&0\\0&b_i
\end{pmatrix}
\quad \mathrm{and}\quad
K_j=\begin{pmatrix}
0&d_j\\c_j&0
\end{pmatrix}.
\end{equation}
The transformable states via IO are
\begin{eqnarray}
\Lambda^{IO}(\rho)&=&
\sum_iK_i\rho K_i^\dag+\sum_jK_j\rho K_j^\dag\n
&=&\frac{1}{2}\begin{pmatrix}
\sum_i|a_i|^2(1+z)+\sum_j|d_j|^2(1-z)&(\sum_ia_ib_i^*+\sum_jd_jc_j^*)r\\
(\sum_ib_ia_i^*+\sum_jc_jd_j^*)r&\sum_i|b_i|^2(1-z)+\sum_j|c_j|^2(1+z)
\end{pmatrix}.
\end{eqnarray}
In other words, the transformable range $(z',$ $r')$ represented in the Bloch sphere is given by
\begin{numcases}{}
r'=gr,\\
\label{constraint1}z'=\sum_i|a_i|^2(1+z)+\sum_j|d_j|^2(1-z)-1,
\end{numcases}
with $g=\sum_ia_ib_i^*+\sum_jd_jc_j^*$.
Another constraint is
\begin{equation}\label{constraint2}
\sum_i|a_i|^2+\sum_j|c_j|^2=\sum_i|b_i|^2+\sum_j|d_j|^2=1,
\end{equation}
due to the condition of $\sum_iK_i^\dag K_i+\sum_jK_j^\dag K_j=I$.
By choose suitable phases for $a_i$, $b_i$, $c_j$, and $d_j$,
we can get
\begin{equation}
|g|=\sum_i|a_i|\cdot|b_i|+\sum_j|d_j|\cdot|c_j|.
\end{equation}

Now we use the Lagrangian multiplier method to calculate
the extremum of $|g|$ under the constraints of Eqs. (\ref{constraint1}) and (\ref{constraint2}).
Define Lagrangian function
$G=G(|a_i|,$ $|b_i|,$ $|c_j|,$ $|d_j|,$ $\lambda_1,$ $\lambda_2,$ $\lambda_3)$ as the following form:
\begin{eqnarray}
G&=&|g|+\lambda_1[\sum_i|a_i|^2(1+z)+\sum_j|d_j|^2(1-z)-(1+z')]\n
&&+\lambda_2(\sum_i|a_i|^2+\sum_j|c_j|^2-1)
+\lambda_3(\sum_i|b_i|^2+\sum_j|d_j|^2-1).
\end{eqnarray}
At the extreme point, the partial derivatives of $G$ are equal to zero,
and then we obtain that
\begin{subequations}
\begin{numcases}{}
\label{constrainta}|a_i|=-2|b_i|\lambda_3,\\
\label{constraintb}|d_j|=-2|c_j|\lambda_2,\\
\label{constraintc}\sum_i|a_i|^2+\sum_j|c_j|^2=1,\\
\label{constraintd}\sum_i|b_i|^2+\sum_j|d_j|^2=1,\\
\label{constrainte}4\lambda_3=\frac{1}{\lambda_1(1+z)+\lambda_2},\\
\label{constraintf}4\lambda_2=\frac{1}{\lambda_1(1-z)+\lambda_3},\\
\label{constraintg}z'=\sum_i|a_i|^2(1+z)+\sum_j|d_j|^2(1-z)-1.
\end{numcases}
\end{subequations}
According to Eq. (\ref{constrainta}), (\ref{constraintb}), (\ref{constraintc}), and (\ref{constraintd}),
we have
\begin{eqnarray}\label{solution(abcd)}
\left\{
\begin{array}{l}
\sum_i|a_i|^2=\frac{4\lambda_3^2(1-4\lambda_2^2)}{1-16\lambda_2^2\lambda_3^2},\\
\sum_i|b_i|^2=\frac{1-4\lambda_2^2}{1-16\lambda_2^2\lambda_3^2},\\
\sum_j|c_j|^2=\frac{1-4\lambda_3^2}{1-16\lambda_2^2\lambda_3^2},\\
\sum_j|d_j|^2=\frac{4\lambda_2^2(1-4\lambda_3^2)}{1-16\lambda_2^2\lambda_3^2}.
\end{array}
\right.
\end{eqnarray}
Solving Eqs. (\ref{constrainte}) and (\ref{constraintf}),
we get
\begin{eqnarray}\label{lambda12}
\lambda_3=-\frac{1-z}{2}(\lambda_1+\sqrt{\lambda_1^2+\frac{1}{1-z^2}}),\n
\lambda_2=-\frac{1+z}{2}(\lambda_1+\sqrt{\lambda_1^2+\frac{1}{1-z^2}}).
\end{eqnarray}
Note that $\lambda_3$ should not be greater than zero due to Eq. (\ref{constrainta}), and the solution of Eqs. (\ref{constrainta}), (\ref{constraintb}), (\ref{constraintc}), and (\ref{constraintd})
does not exist if we choose $\lambda_1=0$ as the solution of Eqs. (\ref{constrainte}) and (\ref{constraintf}).
By substituting Eqs. (\ref{solution(abcd)}) and (\ref{lambda12}) into Eq. (\ref{constraintg}),
we have
\begin{equation}\label{Eq-k}
\kappa^2(1-z^2)^2(1-z')-2\kappa(1-z^2)+1+z'=0,
\end{equation}
where $\kappa=[\lambda_1+\sqrt{\lambda_1^2+1/(1-z^2)}]^2$.
The solutions of Eq. (\ref{Eq-k}) are $\kappa_1=(1+z')/[(1-z^2)(1-z')]$ or $\kappa_2=1/(1-z^2)$.
Since $\lambda_1\neq0$, we choose $\kappa=\kappa_1$.
By using the above results,  the extremum of $|g|$ is
\begin{equation}
|g|_{opt}=(1-z')\sqrt{\kappa}=\sqrt{\frac{1-z'^2}{1-z^2}},
\end{equation}
which means that the transformable range by using IO is limited by
\begin{equation}\label{region1}
|r'|=|gr|\leq|g|_{opt}|r|=\sqrt{\frac{1-z'^2}{1-z^2}}|r|.
\end{equation}
The above equation can be rewritten as
\begin{equation}
\frac{z'^2}{1}+(1-z^2)\frac{r'^2}{r^2}\leq1,
\end{equation}
which is just a part of the boundary of transformation region
calculated from a special kind of IO (\ref{Kraus-special}).
Therefore, the maximal transformation region of initial qubit $(z,$ $r)$ via IO
is given by Eq. (\ref{main result}).
\hfill $\square$
\vspace{1ex}


\section*{Acknowledgements}
We thank J.-X. Hou, Y.-Z. Liu, and Y.-H Shi for their valuable discussions.
This work was supported by the NSFC (Grant No.11375141, No.11425522, No.91536108, No.11647057, and No.11705146),
the special research funds of shaanxi province department of education (No.203010005),
Northwest University scientific research funds (No.338020004)
and the double first-class university construction project of Northwest
University.
\section*{Author contributions}

H.-L. Shi and X.-H. Wang initiated the research project and established the main results.
W.-L. Yang, H. Fan, Z.-Y. Yang, and S.-Y. Liu joined some discussions and provided suggestions.
H.-L. Shi and X.-H. Wang wrote the manuscript with advice from W.-L. Yang, H. Fan, Z.-Y. Yang, and S.-Y. Liu.

\section*{Additional information}

\textbf{Competing financial interests:} The authors declare no competing financial interests.

\end{document}